\documentclass[journal]{IEEEtran}
\usepackage{amsmath}
\usepackage{amssymb}
\usepackage{amsthm}
\usepackage{bm}
\usepackage{graphicx}

\usepackage{xcolor}
\usepackage[switch]{lineno}   

\newcommand{\jj}{\mathrm{j}}
\newcommand{\eu}{\mathrm{e}}
\newcommand{\CN}{\mathcal{CN}}
\newcommand{\Grm}{\mathcal{G}}
\newcommand{\Grmb}{\mathcal{G}_{b}}
\newcommand{\inner}[2]{\langle #1,\, #2\rangle}
\newcommand{\Pperp}[1]{\bm{P}^{\perp}_{#1}}
\newcommand{\Tpri}{T_{\mathrm{PRI}}}
\newcommand{\Reff}{R_0}
\newcommand{\crb}{\mathrm{CRB}}

\newcommand{\vg}{\bm{g}}        
\newcommand{\vl}{\bm{\ell}}     
\newcommand{\vt}{\bm{t}}        
\newcommand{\vh}{\bm{h}}        
\newcommand{\vfh}{\bm{f}_{\!H}} 
\newcommand{\vone}{\bm{1}}

\newtheorem{corollary}{Corollary}
\newtheorem{proposition}{Proposition}

\begin{document}
%\linenumbers

\title{Velocity Information Geometry\\ of Coherent Intra-CPI Waveform Agility}

\author{Charles E. Thornton \\ Virginia Tech NSI, Blacksburg, VA
\thanks{Correspondence: cthorn14@vt.edu}}

\maketitle

\begin{abstract}
Spectrum sharing forces radars to vary carrier frequency and bandwidth on a pulse-to-pulse basis within a coherent processing interval (CPI). While the resulting range-Doppler distortion is well-studied, the corresponding velocity-estimation limit is not. We show that in the resolved-bin slow-time model of coherent agile-CPI processing, the effective Fisher information for radial velocity is the SNR-weighted energy of the carrier--time lever arm that survives projection out of the range and phase nuisance subspace. The carrier sequence thus sets the projection geometry, while the bandwidth sequence enters only through SNR weighting. Two consequences follow. First, the carrier sequence inflates the bound by a closed-form factor governed by the correlation between carrier offset and slow time: randomized or orthogonalized hops are nearly harmless, while ramp-correlated hops can severely degrade velocity information. Second, under matched filtering at equal pulse energy, the velocity Cram\'er--Rao bound (CRB) is invariant to the bandwidth sequence; a corollary recasts the output-SNR loss of agile-CPI mismatched filtering as a processing cost entering only through a per-pulse mismatch loss. The bound is verified against a brute-force Fisher matrix and Monte-Carlo maximum-likelihood estimation. The result yields a design principle: carrier hopping should be chosen not only for spectral coexistence but also to preserve the velocity-information residual.
\end{abstract}


%%%%%%%%%%%%%%%%%%%%%%%%
\section{Introduction}
%%%%%%%%%%%%%%%%%%%%%%%%
Spectrum congestion has made dynamic spectrum sharing a practical requirement for radar, and a common response is to hop carrier frequency and bandwidth on a pulse-to-pulse basis within a single
coherent processing interval (CPI) \cite{martone,kirk,mattingly}. Such intra-CPI waveform agility breaks the assumption that every pulse in the CPI is identical, and a substantial body of work has developed processing techniques to repair the resulting distortion, including
Richardson--Lucy deconvolution \cite{kirk}, phase correction for frequency-agile processing \cite{mattingly}, and mismatched filtering and clutter-compensation approaches for waveform-agile CPIs \cite{mattingly,owen}. Related work has imposed constraints on waveform adaptation to mitigate coherent-processing degradation in agile CPIs \cite{thornton}.

A complementary question has received less attention: what does the waveform agility pattern cost in \emph{velocity information}? Cram\'er--Rao bounds (CRBs) exist for related problems such as
frequency estimation for coherent pulse trains~\cite{pourhomayoun} and velocity / range-shift bounds for stepped-frequency profiling~\cite{liutsp}. However, neither yields a closed-form
velocity bound as a function of the intra-CPI agility pattern under
coherent slow-time processing, and the cognitive random-stepped-frequency literature~\cite{huang} optimizes carrier selection against a sparse-recovery bound, a different objective.

This correspondence shows that velocity information in the resolved-bin coherent agile-CPI model is geometrically determined by a projection residual: the effective Fisher information is the SNR-weighted energy of a carrier--time lever arm after projection out of the range and phase nuisance subspace, so carrier selection sets the information geometry while bandwidth enters only through SNR weighting. Three consequences follow. First, a carrier-hop inflation law (Proposition~\ref{prop:carrier}) shows that velocity-information loss is governed by hop shape rather than hop magnitude: randomized or orthogonalized hop sequences are nearly harmless, whereas ramp-correlated sequences can severely degrade velocity information. Second, bandwidth invariance holds under matched filtering (Proposition~\ref{prop:invariance}). Third, the agile-CPI mismatched-filter loss is identified as a processing rather than fundamental cost (Corollary~\ref{cor:cost}).

%%%%%%%%%%%%%%%%%%%%%%%%%%%%%%
\section{Signal model and velocity information geometry}
%%%%%%%%%%%%%%%%%%%%%%%%%%%%%%
A CPI of $M$ linear frequency-modulated (LFM) pulses, indexed $m=0,\dots,M-1$, is transmitted. Pulse $m$ has total carrier $G(m)=F_0+F_H(m)$ and instantaneous bandwidth $\beta(m)$. We consider a single point target in white noise, with no clutter or interference (the extension to correlated disturbances is discussed in Sec.~\ref{sec:open}). Target motion is purely radial and far-field over the CPI, with tangential motion and wavefront curvature neglected. The baseband, pulse-compressed slow-time sample from the selected target range bin is
\begin{align}
  s(m) &= a\,q_m\,\eu^{\jj\Phi(m)}, \quad
  \Phi(m) \triangleq \phi_0+\phi(m), \\
  \phi(m) &= -\tfrac{4\pi}{c}G(m)(\Reff + v\,t_m),
  \label{eq:slowtime}
\end{align}
where $t_m=m\Tpri$, $v$ is the radial velocity, $A=a\,\eu^{\jj\phi_0}$ is the complex reflectivity, and $q_m>0$ is the known per-pulse peak filter gain. The parameter $\Reff$ is the residual range offset relative to the selected range-bin reference; the coarse range is assumed to have been fixed by pulse compression and range-bin selection. The resulting slow-time sequence contains one complex sample per pulse from the selected bin.

The velocity-dependent term in \eqref{eq:slowtime} is the
slow-time accumulation of the standard Doppler shift. Specifically,
with $f_D(m)=2vG(m)/c$, the phase accrued by pulse $m$ is
$-2\pi f_D(m)t_m=-\tfrac{4\pi}{c}G(m)vt_m$. Intra-pulse Doppler
distortion of the pulse envelope is neglected under the stop-and-hop
approximation below.

Noisy observations are given by $y(m)=s(m)+w(m)$ with $w(m)\sim\CN(0,\sigma^2(m))$, and the per-pulse
SNR is
\begin{equation}
  \gamma(m) \triangleq \frac{a^2\,q_m^2}{\sigma^2(m)}.
  \label{eq:snr}
\end{equation}
We adopt the stop-and-hop approximation, valid when $\max_m f_D(m)/\beta(m)\ll 1$ and $T_p/(M\Tpri)\ll 1$, where $f_D(m) = 2vG(m)/c$ is the per-pulse Doppler shift and $T_p$ the pulse duration. In Section~\ref{sec:open}, we discuss the consequences of relaxing this.

The estimand is $v$, and the nuisance parameters are $(\Reff,\phi_0)$. The amplitude $a$ decouples because its sensitivity $\partial_a s(m)=q_m\eu^{\jj\Phi(m)}$ is in phase quadrature with the phase sensitivities $\partial_\theta s(m)=\jj s(m)\partial_\theta\Phi(m)$, with $\theta\in\{v,\Reff,\phi_0\}$. Hence every
amplitude--phase Fisher-information entry $\operatorname{Re}\{(\partial_a s)^*\partial_\theta s\}$ vanishes. The Fisher information matrix is therefore block-diagonal between $a$ and $(v,\Reff,\phi_0)$, so $a$ may be omitted from the calculation \cite[Ch. 15]{kay}.

Writing $\eta\triangleq 4\pi/c$, the phase sensitivities
\begin{align}
  \partial_v\Phi(m)        &= -\eta\,L(m), \quad L(m)\triangleq G(m)\,t_m,
    \label{eq:dv}\\
  \partial_{\Reff}\Phi(m)  &= -\eta\,G(m), \label{eq:dR}\\
  \partial_{\phi_0}\Phi(m) &= 1 \label{eq:dphi}
\end{align}
define three sensitivity directions: the carrier vector $\vg=[G(0),\dots,G(M-1)]^{\mathsf T}$ (range direction),
the all-ones vector $\vone$ (phase direction), and the ``lever arm"
$\vl=\vg \odot \vt$ (velocity direction), where $\vt\triangleq[t_0,\dots,t_{M-1}]^{\mathsf T}$.
Each entry $L(m)=G(m)t_m$ is the coefficient multiplying $v$ in the phase. Thus, carrier frequency and elapsed slow time jointly determine the velocity sensitivity. The estimation problem therefore reduces to the geometry of three vectors in an SNR-weighted vector space.

For $y(m) = s(m) + w(m)$ with deterministic $s(m)$ and independent
$w(m) \sim \CN(0,\sigma^2(m))$, the Fisher information matrix (FIM) for real
parameters $\theta$ is $[\bm{J}(\theta)]_{ij} = \sum_m
\frac{2}{\sigma^2(m)}\,\operatorname{Re}\{\partial_{\theta_i} s^*(m)\,
\partial_{\theta_j} s(m)\}$ \cite{kay}. For the phase parameters we have
$\partial_\theta s(m) = \jj\, s(m)\, \partial_\theta \Phi(m)$ with
$|s(m)|^2 = a^2 q_m^2$, so
\[
  [\bm{J}]_{ij} = 2 \sum_{m=0}^{M-1} \gamma(m)\,
  \partial_{\theta_i}\Phi(m)\, \partial_{\theta_j}\Phi(m).
\]
These entries are naturally expressed through the SNR-weighted inner
product and associated Gram matrix
\begin{multline}
  \inner{\bm{x}}{\bm{y}} \triangleq \sum_{m=0}^{M-1}\gamma(m)\,x(m)\,y(m),
  \\
  [\Grm(\bm{u}_1,\dots,\bm{u}_k)]_{ij} = \inner{\bm{u}_i}{\bm{u}_j},
  \label{eq:gram}
\end{multline}
so that $[\bm{J}]_{ij} = 2\inner{\partial_{\theta_i}\Phi}{\partial_{\theta_j}\Phi}$:
each FIM entry is the SNR-weighted inner product of the corresponding
sensitivity vectors. With the sensitivities \eqref{eq:dv}--\eqref{eq:dphi}, the FIM in $(v, \Reff, \phi_0)$ is
\[
  \bm{J} = 2
  \begin{bmatrix}
    \eta^2 \inner{\vl}{\vl} & \eta^2 \inner{\vl}{\vg} & -\eta \inner{\vl}{\vone} \\
    \eta^2 \inner{\vl}{\vg} & \eta^2 \inner{\vg}{\vg} & -\eta \inner{\vg}{\vone} \\
    -\eta \inner{\vl}{\vone} & -\eta \inner{\vg}{\vone} & \inner{\vone}{\vone}
  \end{bmatrix},
\]
and the effective information for $v$ is the Schur complement over the nuisance block. The sign and $\eta$ factors on the nuisance sensitivities cancel in the complement (they rescale the nuisance block and its cross terms consistently), leaving
\[
  I_v^{\mathrm{eff}}
  = 2\eta^2 \Bigl( \inner{\vl}{\vl}
  - \bm{b}^{\mathsf{T}} \Grm(\vg,\vone)^{-1} \bm{b} \Bigr),
  \qquad
  \bm{b} = \bigl[\inner{\vl}{\vg}, \; \inner{\vl}{\vone}\bigr]^{\mathsf{T}}.
\]
The bracketed quantity is the squared $\gamma$-distance from $\vl$ to $\operatorname{span}\{\vg,\vone\}$ \cite{scharf}. Eliminating the nuisance parameters $(\Reff,\phi_0)$ by Schur complement is equivalent to projecting $\vl$ onto the $\gamma$-orthogonal complement of $\mathrm{span}\{\vg,\vone\}$. Therefore the effective FIM for $v$ takes a transparent geometric form: it is the squared SNR-weighted norm of the
lever arm after projecting out the nuisance subspace
$\mathrm{span}\{\vg,\vone\}$, written as
\begin{equation}
  I_v^{\mathrm{eff}}
  = 2\eta^2\,\bigl\|\Pperp{\{\vg,\vone\}}\vl\bigr\|_\gamma^2,
  \quad \|\bm{x}\|_\gamma^2 \triangleq \inner{\bm{x}}{\bm{x}},
  \label{eq:projform}
\end{equation}
where $\Pperp{\{\vg,\vone\}}$ is the $\gamma$-weighted orthogonal
projector onto the complement of $\mathrm{span}\{\vg,\vone\}$, given explicitly by
\[
  \Pperp{\{\vg,\vone\}}
  = \bm{I}_M - \bm{U}\left(\bm{U}^{\mathsf{T}}\bm{\Gamma}\bm{U}\right)^{\dagger}\bm{U}^{\mathsf{T}}\bm{\Gamma},
  \quad
  \bm{U} = [\vg \;\, \vone],
\]
with $\bm{\Gamma} = \operatorname{diag}\{\gamma(0),\dots,\gamma(M-1)\}$ and $(\cdot)^{\dagger}$ the Moore--Penrose pseudoinverse, which reduces to the ordinary inverse when $\vg$ and $\vone$ are linearly independent.\footnote{Equivalently, $\|\Pperp{\{\vg,\vone\}}\vl\|_\gamma^2$ may be computed as the squared Euclidean norm of the ordinary orthogonal projection of $\bm{\Gamma}^{1/2}\vl$ onto the complement of $\operatorname{span}\{\bm{\Gamma}^{1/2}\vg,\bm{\Gamma}^{1/2}\vone\}$. In the whitened coordinates $\tilde{\bm{x}}=\bm{\Gamma}^{1/2}\bm{x}$, Euclidean orthogonality is $\gamma$-orthogonality among the originals.} Since $\vone$ spans part of the nuisance subspace, translating the slow-time origin shifts $\vl$ only by a multiple of $\vg$, which also lies in the nuisance subspace; the effective information \eqref{eq:projform} and the resulting CRB are therefore invariant to the choice of slow-time origin.

In computable form, the same Schur complement is a ratio of Gram determinants,
\begin{equation}
  I_v^{\mathrm{eff}}
  = 2\eta^2\,\frac{\det\Grm(\vl,\vg,\vone)}{\det\Grm(\vg,\vone)},
  \label{eq:Ieff}
\end{equation}
and hence
\begin{equation}
  \crb(v) \;=\; \frac{c^2}{32\pi^2}\,
  \frac{\det\Grm(\vg,\vone)}{\det\Grm(\vl,\vg,\vone)}.
  \label{eq:crb}
\end{equation}

The determinant form \eqref{eq:Ieff}--\eqref{eq:crb} assumes $\det\Grm(\vg,\vone)\neq 0$, i.e., that $\vg$ and $\vone$ are linearly independent; in the uniform-carrier case the ratio is formally $0/0$ (Sec.~\ref{sec:uniform}). The projector form \eqref{eq:projform}, computed with the pseudoinverse above, remains well-defined in the degenerate case and recovers the reduction of Sec.~\ref{sec:uniform}.

Equation \eqref{eq:projform} is the central result of this correspondence. It shows that velocity information is precisely the component of the carrier--time lever arm that remains after removing the range and phase nuisance directions. The carriers $\{G(m)\}$ determine the geometry through the vectors $\vg$ and $\vl$, whereas the bandwidths ${\beta(m)}$ influence the bound only through the weights $\gamma(m)$.

%%%%%%%%%%%%%%%%%%%%%%%%%%%%%%%%%%%%%%%%%%%%%%
\section{Bandwidth invariance and carrier-hop dependence}
\label{sec:invariance}
%%%%%%%%%%%%%%%%%%%%%%%%%%%%%%%%%%%%%%%%%%%%%%

\begin{proposition}[Bandwidth invariance in the resolved-bin matched-filter model]
\label{prop:invariance}
For a single point target in white noise, within the resolved-bin slow-time model \eqref{eq:slowtime}, under per-pulse matched filtering at equal pulse energy, and within the stop-and-hop approximation, the velocity CRB \eqref{eq:crb} is independent of the intra-CPI bandwidth
sequence $\{\beta(m)\}$.
\end{proposition}
\begin{proof}[Sketch]
$\beta(m)$ appears nowhere in $\vg$, $\vl$, or $\vone$. Under matched
filtering at equal pulse energy the compressed peak SNR is independent
of bandwidth, so $\gamma(m)$ is constant in $m$; $\beta(m)$ therefore
enters no Gram entry in \eqref{eq:crb}. The invariance is exact
(verified to machine precision, Sec.~\ref{sec:numerics}).
\end{proof}

Within this model, $\beta(m)$ enters only through the compressed pulse shape and the per-pulse gain $q_m$ (Sec.~V), never through the phase history; the content of Proposition~\ref{prop:invariance} is thus the precise delineation of the assumptions under which bandwidth agility carries no fundamental velocity information, in contrast to the processing-induced losses of Corollary~\ref{cor:cost}. Section~\ref{sec:open} identifies relaxation of the stop-and-hop approximation as the mechanism by which bandwidth re-enters the velocity sensitivity.

For the carrier channel, write $\vg=F_0\vone+\vfh$, so that
$\mathrm{span}\{\vg,\vone\}=\mathrm{span}\{\vone,\vfh\}$ for any
non-constant $\vfh$. Let $F_{H,\mathrm{rms}}=\|\vfh\|_2/\sqrt{M}$ denote the
root-mean-square hop magnitude, $\vh\triangleq \vfh/F_{H,\mathrm{rms}}$
the unit-RMS hop shape, and
$\epsilon\triangleq F_{H,\mathrm{rms}}/F_0$ the (dimensionless, scalar) hop amplitude, $F_0$ being the scalar center carrier.
Let
\begin{equation}
  r_\gamma \;=\;
  \frac{\inner{\Pperp{\vone} \vt}{\Pperp{\vone} \vfh}_\gamma}
       {\|\Pperp{\vone} \vt\|_\gamma\,\|\Pperp{\vone} \vfh\|_\gamma}
  \label{eq:rgamma}
\end{equation}
be the SNR-weighted correlation between centered slow time and
centered carrier offset.

\begin{proposition}[Carrier-hop dependence]
\label{prop:carrier}
Under the conditions of Proposition~\ref{prop:invariance} with a
non-constant carrier sequence,
\begin{equation}
  \frac{\crb_{\mathrm{agile}}(v)}{\crb_{\mathrm{unif}}(v)}
  = [(1-r_\gamma^2) + O(\epsilon)]^{-1},
  \label{eq:carrier_ratio}
\end{equation}
where $\crb_{\mathrm{unif}}$ is the uniform-CPI bound, derived in Sec.~\ref{sec:uniform}. When $1-r_\gamma^2$ is bounded away from zero this reduces to $(1-r_\gamma^2)^{-1}[1+O(\epsilon)]$, controlled by the hop shape $r_\gamma$ and not the amplitude $\epsilon$. In the
pure-ramp limit $F_H(m)\propto t_m$ the leading term vanishes, the surviving information is second order in $\epsilon$, and the inflation scales as $O(\epsilon^{-2})$. 
\end{proposition}

\begin{proof}
By \eqref{eq:projform}, $I_v^{\mathrm{agile}}
=2\eta^2\|\Pperp{\{\vone,\vfh\}}\vl\|_\gamma^2$ with
$\vl=F_0(\vt+\epsilon\,\vh\odot \vt)$. Hence
\begin{equation*}
  \tfrac{1}{2\eta^2 F_0^2}\,I_v^{\mathrm{agile}}
  = \|\Pperp{\{\vone,\vh\}}\vt\|_\gamma^2 + O(\epsilon),
\end{equation*}
the $O(\epsilon)$ collecting cross terms in $\vh\odot \vt$. Writing
$\vt_c=\Pperp{\vone}\vt$, $\vh_c=\Pperp{\vone}\vh$,
\begin{equation*}
  \|\Pperp{\{\vone,\vh\}}\vt\|_\gamma^2
  = \|\vt_c\|_\gamma^2 - \frac{\inner{\vt_c}{\vh_c}_\gamma^2}{\|\vh_c\|_\gamma^2}
  = \|\vt_c\|_\gamma^2(1-r_\gamma^2).
\end{equation*}
Since $I_v^{\mathrm{unif}}=2\eta^2 F_0^2\|\vt_c\|_\gamma^2$
(from the CRB in Sec.~\ref{sec:uniform}), dividing gives \eqref{eq:carrier_ratio}.
\end{proof}

The pure-ramp scaling follows directly from \eqref{eq:projform}. When $\vh$ is affine in $\vt$, the nuisance subspace $\operatorname{span}\{\vone,\vh\}$ contains $\vt$, so $\Pperp{\{\vone,\vh\}}\vt=\bm{0}$ and only the cross term of $\vl=F_0(\vt+\epsilon\,\vh\odot\vt)$ survives the projection:
\[
  I_{v}^{\mathrm{agile}}
  = 2\eta^2F_0^2\epsilon^2\,
  \bigl\|\Pperp{\{\vone,\vh\}}(\vh\odot\vt)\bigr\|_\gamma^2
  = O(\epsilon^2),
\]
provided the projected quadratic component is nonzero, and hence $\crb(v)=O(\epsilon^{-2})$. The $O(\epsilon)$ residual in \eqref{eq:carrier_ratio} is sign-indefinite. The terms ``inflation'' and ``cost'' attach to the nonnegative leading factor $1/(1-r_\gamma^2)\ge 1$ (equality iff $r_\gamma=0$); the cross terms at $O(\epsilon)$ may take either sign, and the orthogonalized sequence of Table~\ref{tab:carrier} ($r_\gamma=0$, ratio $0.9998$) realizes a marginally favorable instance. A well-chosen hop sequence thus improves on the uniform CPI by at most an $O(\epsilon)$ margin, never at leading order.

Two points deserve examination. First, away from the near-ramp regime the inflation is controlled by the shape $r_\gamma$, not the amplitude $\epsilon$. Shrinking the hops in amplitude does not
reduce $r_\gamma$, so the multiplicative correction is $O(\epsilon)$ relative to a finite $1/(1-r_\gamma^2)$. This corrects a natural but mistaken intuition that front-end-limited hops ($\epsilon\!\sim\!1\%$) must be harmless. As $r_\gamma\!\to\!1$ the expansion is no longer uniform: the leading $1-r_\gamma^2$ and the $O(\epsilon)$ term become comparable, and in the pure-ramp limit the amplitude controls the bound.

Second, the cost has a clean geometric interpretation: a non-constant carrier renders fine range observable, and \eqref{eq:carrier_ratio} is the velocity-information price of the nuisance dimension spent on that observability, a price that vanishes only when the hop pattern is uncorrelated with slow time. Proposition~\ref{prop:carrier} thus doubles as a design rule: randomized or explicitly orthogonalized hop sequences ($r_\gamma\!\approx\!0$) preserve velocity accuracy while ramp-correlated sequences should be avoided.

\begin{corollary}[Randomized hopping]
\label{cor:random}
For hop offsets drawn i.i.d.\ and independent of slow time,
$r_\gamma=O_p(M^{-1/2})$, so the CRB inflation is $1+O_p(M^{-1})$, which is
near-optimal in the typical sense, incurring a penalty that vanishes as the
CPI lengthens.
\end{corollary}

\noindent\emph{Relation to the stepped-frequency projection bound.}
The projection-residual form \eqref{eq:projform} shares its mathematical structure with the stepped-frequency velocity CRB of Liu \emph{et al.}~\cite{liutsp}. The models, however, differ substantially. Liu \emph{et al.} retain both inter-pulse phase and envelope-delay information in a full-waveform high-range-resolution profiling framework, whereas the present model begins after pulse compression and range-bin selection and therefore retains only the resolved-bin coherent slow-time phase history. This reduction collapses the nuisance geometry to the explicit subspace $\mathrm{span}\{\vg,\vone\}$ and exposes the closed-form hop-correlation penalty \eqref{eq:carrier_ratio}, which is not explicit in the full-waveform formulation.

%%%%%%%%%%%%%%%%%%%%%%%%%%%%%%%%%%
\section{Uniform-CPI reduction}
\label{sec:uniform}
%%%%%%%%%%%%%%%%%%%%%%%%%%%%%%%%%%

Set $G(m)=F_0$, $\gamma(m)=\gamma$. Then $\vg\parallel\vone$, so
$\det\Grm(\vg,\vone)=0$, and \eqref{eq:crb} is nominally $0/0$. In other words, range
and absolute phase are confounded when the carrier is constant.
Collapsing the two confounded nuisances into a single constant phase
gives the proper $2\times2$ reduction
\begin{equation}
  \crb(v) = \frac{3c^2}{8\pi^2 F_0^2\,\gamma\,M(M^2-1)\,\Tpri^2},
  \label{eq:uniform}
\end{equation}
recovering the familiar coherent bound $\crb(v)\propto 1/[F_0^2\,
\mathrm{SNR}\,M^3]$; this is the $\crb_{\mathrm{unif}}$ appearing in \eqref{eq:carrier_ratio}. 
Defining the wavelength
\[
  \lambda=\frac{c}{F_0},
\]
and the conventional Doppler velocity resolution
\[
  \delta v
  =
  \frac{\lambda}{2(M-1)\Tpri},
\]
Eq.~\eqref{eq:uniform} may be written equivalently as
\[
  \crb(v)
  =
  \frac{3(M-1)}
       {2\pi^2\gamma M(M+1)}
  \,\delta v^2,
\]
or, for large $M$,
\[
  \crb(v)
  \approx
  \frac{3}{2\pi^2\gamma M}\,\delta v^2.
\]
Thus, the uniform-CPI velocity CRB is proportional to the squared Doppler velocity resolution, with the proportionality factor determined by the coherent integration length and per-pulse SNR.

The degeneracy is itself informative: intra-CPI frequency agility is precisely what renders range separable from absolute phase.

%%%%%%%%%%%%%%%%%%%%%%%%%%%%%%%%%%%%%%%%%%%%%%%%%
\section{Per-pulse SNR and mismatched filtering}
%%%%%%%%%%%%%%%%%%%%%%%%%%%%%%%%%%%%%%%%%%%%%%%%%

The bandwidth/filter dependence lives entirely in the per-pulse gain $q_m$ (equivalently the noise variance $\sigma^2(m)$); successive pulses occupy independent PRIs, so the slow-time noise stays diagonal. For a constant-amplitude LFM of fixed duration $T_p$ and sample rate, the pulse energy $E_p$ is independent of $\beta(m)$: a wider sweep at constant amplitude and duration carries the same energy. The matched-filter peak SNR $\gamma_0=2E_p/N_0$ is therefore constant in $m$, and $\beta(m)$ cancels from every Gram entry, which is the mechanism
behind Proposition~\ref{prop:invariance}.

Suppressing range-sidelobe modulation (RSM) across a bandwidth-agile CPI requires reshaping every
pulse's compressed response to a common template. We take the template to be the matched response of the minimum-bandwidth pulse, $Y(\omega)=|X_{\min}(\omega)|^2$, or the common-response choice, which imposes no super-resolution on any pulse and incurs the least loss.
For each frequency $\omega$, the filter is obtained from the
Tikhonov-regularized least-squares problem
\begin{equation*}
  \min_{H(\omega)}
  \left|H(\omega)X_m(\omega)-Y(\omega)\right|^2
  +\delta\left|H(\omega)\right|^2,
\end{equation*}
where $\delta>0$ is the regularization parameter. Its pointwise
minimizer is
\begin{equation}
  \hat H_m(\omega)
  =
  \frac{X_m^{*}(\omega)\,Y(\omega)}
       {|X_m(\omega)|^2+\delta},
  \label{eq:wiener}
\end{equation}
which is the standard Tikhonov solution \cite{tikhonov}.

Let $\bm{x}_m$ denote the sampled baseband waveform of pulse $m$, with
DFT $X_m(\omega)$, and $\hat{\bm{h}}_m=\mathcal{F}^{-1}\{\hat H_m^{*}(\omega)\}$
the filter-coefficient vector, applied as a correlator: the peak filter output
is $\hat{\bm{h}}_m^{\mathsf H}\bm{x}_m$, so the matched filter
$\hat H_m(\omega)=X_m^{*}(\omega)$ corresponds to $\hat{\bm{h}}_m=\bm{x}_m$,
the equality condition in \eqref{eq:mmloss}. The corresponding mismatch loss is
\begin{equation}
  L_{\mathrm{mm}}(m) = \frac{\|\hat{\bm{h}}_m\|^2\,\|\bm{x}_m\|^2}{|\hat{\bm{h}}_m^{\mathsf H}\bm{x}_m|^2} \ge 1
  \quad (\text{unity iff }\hat{\bm{h}}_m\propto \bm{x}_m).
  \label{eq:mmloss}
\end{equation}
For the minimum-bandwidth pulse, \eqref{eq:wiener} reduces to the matched filter in the zero-loading limit, so $L_{\mathrm{mm}}(\beta_{\min})=1$; for a wider pulse the filter retains only the $\beta_{\min}$-wide slice of its spectrum, losing SNR in proportion to the bandwidth excess,
$L_{\mathrm{mm}}(m)\approx \beta(m)/\beta_{\min}$
(shown numerically in Sec.~\ref{sec:numerics}). 

With $\gamma(m)=\gamma_0/L_{\mathrm{mm}}(m)
=\gamma_0\,w(m)$, the scale $\gamma_0$ factors out of every Gram entry
and the bound becomes
\begin{equation}
  \crb(v)
  =
  \frac{c^2}{32\pi^2\,\gamma_0}\,
  \frac{\det\Grmb(\vg,\vone)}
       {\det\Grmb(\vl,\vg,\vone)}.
  \label{eq:crb_weighted}
\end{equation}
Here, $\Grmb$ denotes the Gram matrix induced by the normalized
mismatch weights $w(m)$:
\begin{equation*}
  \bigl[\Grmb(\bm{u}_1,\ldots,\bm{u}_k)\bigr]_{ij}
  \triangleq
  \sum_{m=0}^{M-1}
  w(m)\,u_i(m)\,u_j(m),
\end{equation*}
where $u_i(m)$ denotes the $m$th entry of the vector $\bm{u}_i$. As with \eqref{eq:crb}, the determinant form \eqref{eq:crb_weighted} is degenerate when the carrier is uniform: $\vg=F_0\vone$, so both Gram determinants vanish irrespective of the weights $w(m)$; the degeneracy is one of nuisance geometry, not of the metric. The bandwidth-agile scenario of Sec.~\ref{sec:numerics} is therefore evaluated via the weighted two-parameter reduction, with the confounded pair $(\Reff,\phi_0)$ collapsed to a single phase nuisance and $w(m)$ carried through the reduced Gram matrices $\Grmb(\vone)$ and $\Grmb(\vl,\vone)$ (equivalently, via the pseudoinverse projector of \eqref{eq:projform}).

\begin{corollary}[Cost of RSM mitigation]
\label{cor:cost}
When RSM mitigation forecloses the matched filter and the
minimum-bandwidth-template filter \eqref{eq:wiener} is used instead,
the velocity CRB is inflated relative to Proposition~\ref{prop:invariance} only through the SNR weights $\gamma(m)=\gamma_0\,w(m)$, with $w(m)=1/L_{\mathrm{mm}}(m)$ and $L_{\mathrm{mm}}(m)\approx \beta(m)/\beta_{\min}$. The empirical output-SNR variation reported for agile-CPI filters is a property of the chosen (uniformly weighted) processor; the FIM weights by $w(m)$ automatically.
\end{corollary}

%%%%%%%%%%%%%%%%%%%%%%%%%%%%%%%%%%%%%%%%%%%%%
\section{Numerical evaluation}
\label{sec:numerics}
%%%%%%%%%%%%%%%%%%%%%%%%%%%%%%%%%%%%%%%%%%%%%
We use the parameters of Mattingly \emph{et al.}~\cite{mattingly}:
$F_0=5$~GHz, $\Tpri=200~\mu$s, $T_p=30~\mu$s, $f_s=100$~MHz, $M=256$,
hops within $\pm 50$~MHz of the $100$~MHz front end, and bandwidths
over $6.54$--$53.91$~MHz. The true target parameters are $v=20$~m/s and
$\Reff=0.137$~m at per-pulse SNR $20$~dB, with i.i.d.\ uniform draws
of the hop and bandwidth sequences unless noted, using a fixed seed for
reproducibility. All matched-filter scenarios have constant per-pulse SNR, $\gamma(m)=\gamma_0$, by the mechanism of Proposition~\ref{prop:invariance}; the RSM-control mismatched-filter scenario has $\gamma(m)=\gamma_0/L_{\mathrm{mm}}(m)$ as in Sec.~V. The bound is verified against \eqref{eq:uniform} to machine precision, and independently against a brute-force complex-AWGN FIM assembled
directly from $s(m)=A\eu^{\jj\phi(m)}$, with $8$-digit agreement confirming both the Schur reduction and the constant $(c^2/32\pi^2)$. It is further validated against a Monte-Carlo maximum-likelihood estimator (MLE) which never uses the FIM. The MLE validation uses the uniform-carrier case, for which $\Reff$ and $\phi_0$ are confounded and absorb into a single complex amplitude, profiled out in closed form; the resulting profile-likelihood estimator
\[
  \widehat{v}
  = \underset{v}{\operatorname{arg\,max}}
  \Bigl|\sum_{m=0}^{M-1} y(m)\,\eu^{\jj\frac{4\pi}{c}F_0 v t_m}\Bigr|^2
\]
is maximized by a grid search over $401$ velocities spanning $\pm6\sqrt{\crb(v)}$ about the true value. Over $400$ trials per SNR, the empirical variance of $\widehat{v}$ is consistent with $\crb(v)$ within Monte-Carlo uncertainty ($\operatorname{var}(\widehat{v})/\crb(v) \in [0.88, 1.05]$).

Under matched filtering, a bandwidth-agile CPI gives a velocity CRB
identical to the uniform-CPI value to $12$ digits, confirming
Proposition~\ref{prop:invariance}.

Table~\ref{tab:carrier} sweeps the hop shape at fixed
amplitude, with sequences constructed as
\begin{equation}
  \vh = \sqrt{M} \left[ r\,\frac{\vt_c}{\|\vt_c\|} + \sqrt{1-r^2}\,\frac{\bm{u}_\perp}{\|\bm{u}_\perp\|} \right],
  \label{eq:hopgen}
\end{equation}
with $\vt_c$ as in the proof of Proposition~\ref{prop:carrier} and $\bm{u}_\perp$ a random vector orthogonalized
against $\bm{1}$ and $\vt_c$, then scaled to the hop amplitude. The construction parameter $r$ in \eqref{eq:hopgen} prescribes the unweighted correlation of the hop shape with centered slow time, whereas $r_\gamma$ in \eqref{eq:rgamma} is the realized SNR-weighted correlation; under the constant per-pulse SNR of the matched-filter scenarios the two coincide. The measured inflation matches $1/(1-r_\gamma^2)$ across four orders of magnitude
in the penalty. Hops drawn i.i.d.\ uniform give $r_\gamma=O_p(M^{-1/2})$ (Corollary~\ref{cor:random}) and a ${<}1\%$ effect; an orthogonalized sequence removes the leading-order penalty exactly ($r_\gamma=0$ to machine precision), leaving only the $O(\epsilon)$ residual; a ramp-correlated sequence at $r_\gamma=0.9$ inflates the CRB by
$5.3\times$. Varying $\epsilon$ away from the ramp limit moves only the $O(\epsilon)$ correction, confirming that the cost is set by hop shape, not size; in the pure-ramp case the inflation becomes
amplitude-controlled, $O(\epsilon^{-2})$ (here ${\sim}3\times 10^5$ at $\epsilon=2\times 10^{-3}$).

\begin{table}[t]
\centering
\caption{Carrier-hop velocity-CRB inflation vs.\ hop shape (matched filter, $r_\gamma$ from \eqref{eq:rgamma}). Rows: i.i.d.\ uniform draws ($r_\gamma$ random, typically ${\approx}0$); explicit orthogonalization ($r_\gamma=0$ by construction); sequences from \eqref{eq:hopgen} with prescribed $r=0.5$ and $0.9$ ($r=r_\gamma$ under constant per-pulse SNR); limiting pure ramp ($r_\gamma=1$).}
\label{tab:carrier}
\begin{tabular}{lcc}
\hline
Hop sequence & $r_\gamma$ & $\crb_{\mathrm{agile}}/\crb_{\mathrm{unif}}$ \\
\hline
i.i.d.\ uniform               & $\approx 0$ & $\approx 1.00$ \\
orthogonalized                & $0.000$     & $0.9998$ \\
ramp-correlated               & $0.500$     & $1.333$ \\
ramp-correlated               & $0.900$     & $5.26$ \\
pure ramp ($\vfh\!\propto\!\vt$) & $1.000$     & $\gg 1$$^{\ast}$ \\
\hline
\multicolumn{3}{l}{\footnotesize $^{\ast}$Amplitude-controlled, $O(\epsilon^{-2})$: ${\approx}3\times10^{5}$ at $\epsilon=2\times10^{-3}$.}
\end{tabular}
\end{table}

For the RSM-control mismatched filter (MMF), the mismatch loss tracks the bandwidth ratio, $L_{\mathrm{mm}}(m)\approx \beta(m)/\beta_{\min}$,
anchored at $L_{\mathrm{mm}}(\beta_{\min})=1$ ($0$~dB) and reaching ${\approx}9.2$~dB at $\beta_{\max}/\beta_{\min}=8.24$, varying by $<0.01$~dB over two decades of Tikhonov loading $\delta$. This is a property of this specific template/loading design, not a universal
mismatch-filter law: a different template, loading, filter length, or spectral taper would give a different $L_{\mathrm{mm}}(\beta)$. Propagated through $w(m)=1/L_{\mathrm{mm}}(m)$, a representative
bandwidth-agile pattern inflates the velocity CRB by ${\approx}3.1\times$ (mean per-pulse loss ${\approx}6.4$~dB), i.e.\ $\sqrt{\crb(v)}$ from $1.43$ to $2.49$~mm/s (dashed curve in Fig.~\ref{fig:crb}), which is below the ${\sim}12$~dB ensemble loss reported for the finite-length design in~\cite{mattingly}, consistent with the bandwidth ratio being the large-filter-length floor.

\begin{table}[t]
\centering
\caption{Velocity \(\sqrt{\mathrm{CRB}}\) (mm/s), Sec. VI parameters, 20 dB per-pulse SNR.}
\label{tab:summary}
\begin{tabular}{lc}
\hline
Scenario & $\sqrt{\crb(v)}$ (mm/s) \\
\hline
Uniform CPI & $1.43$ \\
Bandwidth-agile (matched filter) & $1.43$ \\
Frequency-agile, i.i.d.\ hops (matched filter) & $1.43$ \\
Bandwidth-agile (RSM-control MMF) & $2.49$ \\
\hline
\end{tabular}
\end{table}

\begin{figure}[t]
\centering
\includegraphics[width=\columnwidth]{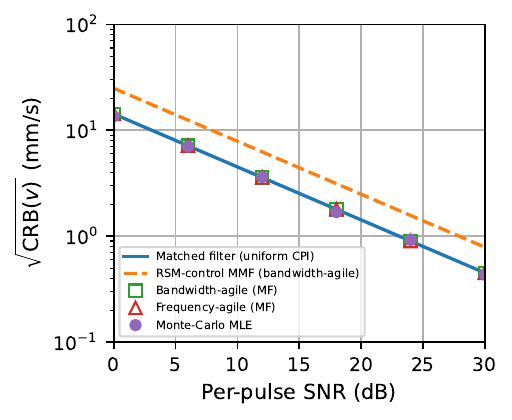}
\caption{Velocity $\sqrt{\crb}$ vs.\ per-pulse SNR. Bandwidth-agile
(any sequence) and frequency-agile (i.i.d.\ hops, $r_\gamma\!\approx\!0$) matched-filter cases (open markers) coincide with the uniform-CPI bound (solid): Proposition~\ref{prop:invariance} and the low-correlation case of Proposition~\ref{prop:carrier}. Monte-Carlo MLE (filled) attains it. The RSM-control mismatched filter (dashed) is the lone offset, the processing cost of
Corollary~\ref{cor:cost}. Ramp-correlated hops (Table~\ref{tab:carrier}) lie well above the solid line.} 
\label{fig:crb}
\end{figure}

%%%%%%%%%%%%%%%%%%%%%%%%%%%%%%%%%%%%%%%%%%%%%%%
\section{Design implication and Conclusions}
\label{sec:open}
%%%%%%%%%%%%%%%%%%%%%%%%%%%%%%%%%%%%%%%%%%%%%%%%

Interpreted as a design constraint, Proposition 2 yields that carrier hopping should be selected not only for spectral coexistence but to preserve the velocity-information residual, since the leading penalty is governed by hop shape rather than hop magnitude. It suffices to keep $|r_\gamma|\le r_{\max}$ for a tolerance set by the allowable inflation $1/(1-r_{\max}^2)$, and the constraint is cheap:
orthogonality to slow time is a single linear constraint on the $M$ hop degrees of freedom; randomized hopping satisfies it typically (Corollary~\ref{cor:random}); explicit orthogonalization satisfies it to machine precision (Table~\ref{tab:carrier}). 

Within the resolved-bin white-noise model, bandwidth agility, by Proposition 1 and Corollary 2, costs velocity information only through the processing choice made to control RSM: a filter-design budget, rather than a waveform-selection criterion.

Two refinements lie beyond this model. The first concerns the stop-and-hop approximation: relaxing it reintroduces the bandwidth into the velocity sensitivity itself, through the range--Doppler coupling $\Delta r(m) = -v\,G(m)\,T_p/\beta(m)$. This is the envelope-delay channel exploited by the full-waveform
stepped-frequency bound~\cite{liutsp} but absent from the resolved-bin phase model, so Proposition~\ref{prop:invariance} is precisely a stop-and-hop statement, with leading corrections governed by
\begin{equation*}
\left( \max_m \frac{f_D(m)}{\beta(m)} \right) \frac{T_p}{M \Tpri}.
\end{equation*}

The second concerns the disturbance. Clutter and interference lie outside the present white noise model: compressed by a different filter on every pulse, their slow-time signature is modulated by the agility pattern, and the disturbance covariance at the resolved bin becomes a full matrix $\bm{C}(\{G(m),\beta(m)\})$. The bound generalizes by taking the projection \eqref{eq:projform} in the $\mathbf{C}^{-1}$ metric, where the invariance of Proposition~\ref{prop:invariance} is expected to break and the design constraint becomes one on the whitened residual, $\|\Pperp{\{\vg,\vone\}}\vl\|_{\bm{C}^{-1}}\ge\kappa$. How a spectrum-dictated pattern degrades that clutter-whitened bound is the natural sequel, with the present result as its baseline.

\end{document}